\documentclass[final,twoside,11pt]{entics} 
\usepackage{enticsmacro}
\usepackage{graphicx}
\usepackage{amsmath,mathtools,bm}
\usepackage{enumerate,relsize,listings,wrapfig}
\usepackage[all]{xy}
\usepackage{cleveref}
\usepackage{pgf,tikz,tikz-qtree,pgfplots}
\usetikzlibrary{arrows,automata,shapes,calc,patterns,intersections}
\pgfplotsset{width=7cm,compat=1.8}
\usepackage{jdist}

\def\conf{MFPS 2023} 	
\def\myconf{mfps39}
\volume{3}			
\def\volu{3}			
\def\papno{12}			
\def\lastname{Kozen, Silva, Voogd} 
\def\runtitle{Joint Distributions for Probabilistic Semantics} 
\def\mydoi{12279}
\def\copynames{D. Kozen, A. Silva, E. Voogd} 
\begin{document}
\begin{frontmatter}
  \title{Joint Distributions in Probabilistic Semantics}

  \author{Dexter Kozen}
  \author{Alexandra Silva}
  \address[a]{Computer Science Department\\ Cornell University\\ Ithaca, New York 14853-7501, USA}
  \author{Erik Voogd}
  \address[b]{Department of Informatics\\ University of Oslo\\ Gaustadalléen 23B, 0373 Oslo, Norway}

\begin{abstract} 
Various categories have been proposed as targets for the denotational semantics
of higher-order probabilistic programming languages.
One such proposal involves joint probability distributions (couplings) used
in Bayesian statistical models with conditioning. In previous treatments, composition of joint measures was performed by disintegrating to
obtain Markov kernels, composing the kernels, then reintegrating to obtain
a joint measure. Disintegrations exist only under certain restrictions on the 
underlying spaces. 
In this paper we propose a category whose morphisms are joint finite measures
in which composition is defined without reference to disintegration, allowing its application to a broader class of spaces. The category is symmetric
monoidal with a pleasing symmetry in which the dagger
structure is a simple transpose.
\end{abstract}
\end{frontmatter}

\section{Introduction}\label{sec:intro}
\emph{Bayesian inference} and \emph{conditioning} are important tools in probabilistic programming. Modern probabilistic languages for machine learning, e.g.~Church \cite{Church08} and Anglican \cite{Anglican15}, generally incorporate these tools in some form. To formalize the semantics of such languages, several authors have proposed categories for modeling Bayesian inference and conditioning via \emph{disintegration}. 

Approaches based on disintegration have the disadvantage that they can only be applied to certain spaces. At the core of the matter is the {\em disintegration theorem} which only holds under certain assumptions---usually that the spaces are standard Borel, as in e.g.~\cite{Pachl78,ChangPollard97,AbramskyBlutePanangaden99,Doberkat07,Panangaden09,Dahlqvist18}. This theorem connects joint distributions and Markov Kernels. A Markov kernel $P$ gives rise to a joint distribution $\J P$ on $X\times Y$ with marginals $\mu$ and $\nu$ in a natural way. The {disintegration theorem} says that this construction can be inverted: for any joint distribution $\theta$ on $X\times Y$ with marginals $\mu$ and $\nu$, there is a kernel $P$, unique up to a $\mu$-nullset, such that $\theta=\J P$. The inversion can be explored in defining composition and proving properties of the underlying category; we will give two examples below. 

In this paper, we want to show that the category $\JDist$ whose objects are measure spaces $(X,\AA,\mu)$ and whose morphisms $(X,\AA,\mu)\to(Y,\BB,\nu)$ are joint distributions $\theta$ on $X\times Y$ with marginals $\mu$ and $\nu$ can be taken as first-class citizen to formalize semantics of probabilistic languages. In particular, we will show that composition in $\JDist$ can be defined without having to use disintegration, thereby making this category more generally applicable than previous approaches. 

One recent proposal by Dahlqvist et al.~\cite{Dahlqvist18} was the category $\Krn$, whose objects are probability spaces $(X,\AA,\mu)$ with $\AA$ a standard Borel space. The morphisms $(X,\AA,\mu)\to(Y,\BB,\nu)$ are equivalence classes modulo $\mu$-nullsets of Markov kernels $P\colon X\to Y$ such that
\begin{align*}
\nu(B) &= \int_{s\in X}P(s,B)\, d\mu(s).
\end{align*}
Dahlqvist et al.~\cite{Dahlqvist18} use disintegration to show that $\Krn$ is a dagger category with involutive functor ${}^\dagger:\Krn\to\Krn^{\mathrm{op}}$. Thus $P^\dagger:Y\to X$ is a kernel in the opposite direction that models Bayesian inference of an input distribution conditioned on output samples. This works even for continuous distributions in which outcomes can occur with zero probability. We will show that the category $\Krn$ can be embedded faithfully in the category $\JDist$, and this is a symmetric monoidal category with joint distributions as tensors and transpose as symmetry: $\theta^\dagger(A\times B) = \theta(B\times A)$. 

Another category was proposed by Abramsky et al.~\cite{AbramskyBlutePanangaden99} who studied a subcategory of $\JDist$ on Polish spaces called $\PRel$, using disintegration to define composition: to compose two joint distributions, one disintegrates to
obtain Markov kernels, composes the kernels as in $\Krn$, then reintegrates to obtain a joint measure. 
This construction can be performed only under various restrictions admitting disintegration \cite{AbramskyBlutePanangaden99,Doberkat07,Panangaden09,CulbertsonSturtz14,ChoJacobs17}. The most common assumption is that the spaces are standard Borel. The most general result of this type seems to be that of Culbertson and Sturtz \cite{CulbertsonSturtz14} based on work of Faden \cite{Faden85}, who assume countably generated spaces but that the measures are perfect. 

In this paper, we show that composition in $\JDist$ can be defined independently without reference to disintegration and without any restriction on the underlying spaces. This allows probabilistic programs to be interpreted in a more general class of spaces with a pleasing symmetry in which the notion of equivalence up to a nullset is built in. We will illustrate our approach in defining the composition in $\JDist$ through an example in \Cref{sec:example}. We will then define the composition in $\JDist$ in 3 steps. We will first present the definition in the discrete setting (\Cref{sec:finite}), as this is a simpler case which does not require extra measure infrastructure. Second, we will define and prove some results akin to \RN\ approximants (\Cref{sec:LRN}). This will provide the necessary results to define composition in $\JDist$ and a functor $\J \colon \Krn \to \JDist$ that provides a faithful embedding of categories (\Cref{sec:JDist}). 

\section{Illustrative Example}\label{sec:example}
\begin{wrapfigure}{r}{.3\textwidth}
\vspace*{-0.5cm}
\hspace*{-1cm}
\begin{lstlisting}[
  xleftmargin=2em,numbersep=2em,
  numbers=left,numberstyle=\tiny\color{darkgray},
  basicstyle=\footnotesize\ttfamily,
  keywordstyle=\color{blue},
  morekeywords={normal,observe,return}]
x := normal(0,1);
y := normal(x,1);
z := normal(y,1);
observe (z > 0.5);
return (x > 1);
\end{lstlisting}
\end{wrapfigure}
\noindent
To explain how composition in $\JDist$ circumvents disintegration, we consider a modified version of the example from Dahlqvist et al. \cite{Dahlqvist18}, pseudocode presented on the right: the goal is to 
measure the probability that $x > 1$ for $x \dist \Norm 01$ given an observation $z>0.5$ for $z \dist \Norm y1$, and $y \dist \Norm x1$.
That is, $x$ is a standard Gaussian sample (Line~1) and $y$ and $z$ are normally distributed with means $x$ and $y$ respectively, and both with variance one (Lines~2 and~3). 
Intuitively, Line~4 conditions on the observation that $z$ is greater than 0.5, and Line~5 returns the probability of $x$ being greater than one given this observation.

To illustrate the nature of the category $\JDist$, we somewhat informally explain the semantics of this program: first in terms of Markov kernels ($\Krn$), and after that in terms of joint distributions ($\JDist$).
\begin{itemize}
  \item In $\Krn$, the state of the program after Line~1 can be thought of as a measure space $(X,\AA,\mu)$ with $(X,\AA) = (\R,\BB_{\R})$ the Borel $\sigma$-algebra on $\R$, and $\mu$ the standard Gaussian measure.
  \item Line~2 constructs a Markov kernel $P : (X,\AA) \to (Y,\BB)$, where $(Y,\BB) = (\R,\BB_\R)$, that, for every $x$, yields the Gaussian measure with variance one centered around $x$.
    The state of $y$ in the program after Line~2 is then the measure space $(Y,\BB,\nu)$ where $\nu(B) := \int_{x \in X} P(x,B)~d\mu$.
  \item Analogously, Line~3 constructs a kernel $Q : (Y,\BB) \to (Z,\CC)$ of Gaussian measures $Q(y,-)$ on $(Z,\CC) = (\R,\BB_\R)$ centered around $y$, and the state of $z$ is the measure space $(Z,\CC,\rho)$ with $\rho(C) := \int_{y \in Y} Q(y,C)~d\nu$.
    Equivalently, $\rho$ is obtained by integration of the kernel composition $P\cmp Q$ w.r.t. $\mu$.
  \item Lines~(4-5) compute an inverse kernel of $P \cmp Q$ as follows: the composition is integrated w.r.t. $\mu$ to obtain a joint measure $\zeta$ on $X \times Z$.
    This joint measure is transposed to $\zeta^\dagger$ on $Z\times X$ and then disintegrated to the inverse kernel $R$.
    Then $\int_{z \in (0.5,\infty)} R(z,-)~d\rho$ evaluated on the interval $(1,\infty)$ yields the desired result.
\end{itemize}
Lines~(1-3) can be pictorially represented using three concrete samples $x,y,z$ as follows:
$$\scalebox{0.85}{$
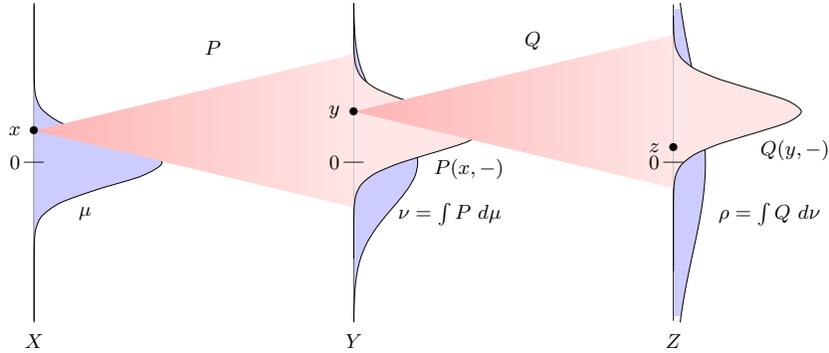
\begin{tikzpicture}[-, >=stealth', node distance=5cm, inner sep=0, auto]
\small
 \node (NW) at (0,0) {};
 \node (N) [right of=NW] {};
 \node (NE) [right of=N] {};
 \node (SW) [below of=NW] {};
 \node (S) [below of=N] {};
 \node (SE) [below of=NE] {};
 \path (NW) edge (SW);
 \path (N) edge (S);
 \path (NE) edge (SE);
 \node (X) [below of=SW, node distance=3mm] {$X$};
 \draw[fill=blue!20,scale=0.5,domain=-5:5,smooth,variable=\t, yshift=-5cm]
 plot ({4*exp(-((\t)/1)^2)},\t);
 \node (mu) at (.8,-3.3) {$\mu$};
 \node (zrL) [below of=NW, node distance=2.5cm, rotate=90] {$|$};
 \node (zrLa) [left of=zrL, node distance=3mm] {$0$};
 \node (Y) [below of=S, node distance=3mm] {$Y$};
 \draw[fill=blue!20,scale=0.5,domain=-5:5,smooth,variable=\t, xshift=10cm, yshift=-5cm]
 plot ({2*exp(-((\t)/2)^2)},\t);
 \node (nu) at (6.5,-3.3) {$\nu = \int P~d\mu$};
 \node (Z) [below of=SE, node distance=3mm] {$Z$};
 \draw[fill=blue!20,scale=0.5,domain=-5:5,smooth,variable=\t, xshift=20cm, yshift=-5cm]
 plot ({1*exp(-((\t)/4)^2)},\t);
 \shade [left color=blue!20, right color=blue!20] (10.008,-0.1) -- (10.11,-0.1) -- (10.11,-4.9) -- (10.008,-4.9);
 \node (rho) at (11.5,-3.3) {$\rho = \int Q~d\nu$};
 \shade [left color=red!30, right color=red!10, opacity=0.5] (0,-2) -- (5,-0.8) -- (5,-3.2) -- (0,-2);
 \node (P) at (2.8,-0.7) {$P$};
 \draw[fill=red!10,scale=0.5,domain=-4:4,smooth,variable=\t, xshift=10cm, yshift=-4cm, opacity=0.5]
 plot ({4*exp(-((\t)/1)^2)},\t);
 \node (Px) at (6.8,-2.6) {$P(x,-)$};
 \shade [left color=red!30, right color=red!10, opacity=0.5] (5,-1.7) -- (10,-0.5) -- (10,-2.9) -- (5,-1.7);
 \node (Q) at (7.8,-0.6) {$Q$};
 \draw[fill=red!10,scale=0.5,domain=-5:3,smooth,variable=\t, xshift=20cm, yshift=-3.4cm, opacity=0.5]
 plot ({4*exp(-((\t)/1)^2)},\t);
 \node (Qx) at (11.9,-2.3) {$Q(y,-)$};
 \node (x) [below of=NW, node distance=2cm] {$\bullet$};
 \node (xa) [left of=x, node distance=3mm] {$x$};
 \node (zrM) [below of=N, node distance=2.5cm, rotate=90] {$|$};
 \node (zrMa) [left of=zrM, node distance=3mm] {$0$};
 \node (y) [below of=N, node distance=1.7cm] {$\bullet$};
 \node (ya) [left of=y, node distance=3mm] {$y$};
 \node (zrE) [below of=NE, node distance=2.5cm, rotate=90] {$|$};
 \node (zrEa) [left of=zrE, node distance=3mm] {$0$};
 \node (z) [below of=NE, node distance=2.27cm] {$\bullet$};
 \node (za) [left of=z, node distance=3mm] {$z$};
\end{tikzpicture}\\[1mm]
\vspace*{-4mm}
$}$$
Interpreting the program in terms of joint distributions can be done as follows:
\begin{itemize}
  \item Line~1 constructs the measure space $(X,\AA,\mu)$ with $\mu$ standard Gaussian as above.
  \item Line~2 constructs a joint measure as follows: let $(Y,\BB)$ and $P$ be as above.
    Define the joint measure $\theta$ on $X\times Y$ on its rectangles as
    $$ \theta (A \times B) := \int_{x \in A} P(x,B)~d\mu $$
    Then the left marginal $\theta(-\times Y)$ of $\theta$ is $\mu$ and the right marginal $\theta(X\times -)$ is $\nu$.
    We consider the state of the program after Line~2 to be the joint distribution $\theta$.
    Intuitively, this is a way to \emph{remember the input distribution}, and this is, computationally speaking, crucial for Bayesian inference.
  \item Similarly, Line~3 can be thought of as constructing a joint $\eta$ on $Y \times Z$ using $(Z,\CC) = (\R,\BB_\R)$ and $Q$ as above.
    This joint will have marginals $\nu$ and $\rho$.
    The joint measures obtained from Lines~(1-3) can be visualized as follows:
$$\scalebox{0.85}{$
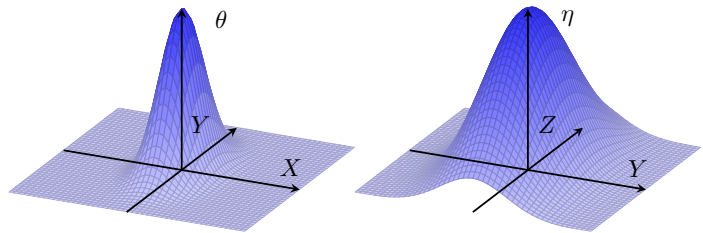
\begin{tikzpicture}
\begin{axis}[ticks=none, xlabel={$X$}, ylabel={$Y$}, zlabel={$\quad\theta$},
  axis lines = center, axis on top, axis line style = {thick},
  colormap={blue}{color=(blue!20) color=(blue!80)}]
\addplot3 [
    domain=-5:5,
    domain y = -5:5,
    samples = 49,
    samples y = 49,
    surf] {4*exp(-x^2)*2*exp(-(y/2)^2)};
\end{axis}
\end{tikzpicture}
\begin{tikzpicture}
\begin{axis}[ticks=none, xlabel={$Y$}, ylabel={$Z$}, zlabel={$\quad\eta$},
  axis lines = center, axis on top, axis line style = {thick},
  colormap={blue}{color=(blue!20) color=(blue!80)}]
\addplot3 [
    domain=-5:5,
    domain y = -5:5,
    samples = 49,
    samples y = 49,
    surf] {2*exp(-(x/2)^2)*exp(-(y/4)^2)};
\end{axis}
\end{tikzpicture}
$}$$
    In order to say anything about the state after Line~3, we have to be able to compose the joint measures $\theta$ and $\eta$ to $\zeta = \theta \cmp \eta$.
    How to define this composition without assumptions on the underlying space is the central contribution of this paper, explained in the following sections.
  \item Lines~(4,5) transpose the joint measure $\zeta$ on $X \times Z$ to $\zeta^\dagger$ on $Z \times X$ and the output of the program is the $\zeta^\dagger$-measure of the rectangle $(0.5,\infty) \times (1,\infty)$.
    \footnote{Computing the result with observe statements on {\em zero measure events}, e.g. $z=0.5$, still requires disintegration: the joint $\zeta^\dagger$ is disintegrated at $z=0.5$ and then evaluated on $(1,\infty)$.}
\end{itemize}
Our goal now is to define how to compose joint measures $\theta$ and $\eta$ (as above) to a joint measure on $X$ and $Z$ such that
\begin{enumerate}
  \item the marginals of $\theta \cmp \eta$ are $\mu$ and $\rho$, and
  \item integration of the kernel composition $P;Q$ with respect to $\mu$ yields $\theta \cmp \eta$.
\end{enumerate}
Before our general treatment of this problem with continuous measures, we describe the problem and its solution in a discrete setting for instructive purposes.

\section{Discrete Approach}
\label{sec:finite}
We consider a finite state space; with the appropriate convergence assumptions, this generalizes straightforwardly to countable state spaces.
Through the below development, we will have forward pointers to the analogous definitions of the general case which we will present in \Cref{sec:notation}, we mark this as ``c.f. $(n)$ below''.

\subsection{Preliminaries}

Formally, a \emph{probability measure} on a finite set $\{1,\ldots,n\}$ (from now on overloadingly written $n$) is a finite additive set function $2^n \to [0,1]$ defined uniquely on its points by $i \mapsto x_i$ (sometimes written $x(i)$) such that $\sum_{i=1}^n x_i = 1$.
Equivalently, it is a vector $x \in [0,1]^n$ such that $\transpose x \cdot \one=1$, where $\one$ denotes the $n$-dimensional vector of ones.
Integrating a function $f : n \to \R$ (always measurable) w.r.t. $x$ is done by weighting each element $f(i)$ with its corresponding probability measure $x_i$.
Thus, integration in the discrete world is just matrix multiplication $\transpose x f$.

A \emph{Markov kernel} is a map $A : n \to [0,1]^n$ (always measurable) such that $A(i)$ is a probability measure for every $i \in n$.
Equivalently, it is a matrix $A \in [0,1]^{n\times n}$ such that $A \one = \one$.
Integrating a Markov kernel $A$ w.r.t. a probability measure $x$ yields a new probability measure $y = \transpose x A$.
This procedure is given componentwise by $y(j) = \sum_{i=1}^n A(i,j) \cdot x(i)$ (c.f. \eqref{eq:measuretransform} below).

A (joint) probability measure on the product space $n \times n$ is a matrix $\alpha \in [0,1]^{n\times n}$ whose entries all sum to one. 
We have $1 = \sum_{i,j} \alpha(i,j) = \transpose\one \alpha \one = \transpose\one \transpose {\alpha} \one$,
so that the row-sums $\alpha\one$ and the column-sums $\transpose {\alpha} \one$ are probability measures; they are the \emph{left} and \emph{right marginal}, respectively.
A joint measure $\alpha$ on $n \times n$ defines unique marginals by definition, but the converse is not true in general: there can be many joint measures for two given marginals.
(In this setting it is solving $2n$ linear equations of $n^2$ unknowns.)

\subsection{Discrete Markov Kernels to Joint Distributions}
We now describe what the functor $\J \colon \Krn \to \JDist$ does to {\em discrete} Markov kernels. 
For a Markov kernel $A$, first define the map $\mem A : n \to [0,1]^{n\times n}$ by (c.f. \eqref{eq:markovlift} below):
$$
\mem A(i,j,k) = 
\begin{cases}
A(i,k) &\text{if }i=j\\
0& \text{otherwise}
\end{cases}
$$ 
 Equivalently, $A'$ is a $n\times n\times n$ matrix where each page (ranged over by $k$) is a diagonal $n\times n$ matrix, and the diagonal of the $k$-th page is given by the $k$-th column of $A$.

Then, we have that $(\transpose x \cdot A')(i,j) = A(i,j) \cdot x_i$ for every $i,j \in n$ and measure $x \in [0,1]^n$ (c.f. \eqref{eq:JPdef} below).
So, each $i$-th row-sum of $\transpose x \cdot A'$ is just $x_i$ (because the row-sums of $A$ are one) and the column-sums are just entry-wise computations of $\transpose A x$. 
This means that $\transpose x \cdot A'$ is a measure on $n\times n$ with left marginal $x$ and right marginal $y=\transpose A x$.
The functor $\J$ maps the kernel $A$ to the joint measure $\transpose x A'$.
It implicitly depends on $x$.

Disintegrating the joint measure $\J A$ back to a kernel is always possible for discrete kernels, but does this give us back $A$?
In the finite setting, to disintegrate a joint measure $\alpha \in [0,1]^{n\times n}$ on the product space $n_1 \times n_2$ (where $n_1 = n_2 = n$), let $\pi : n_1 \times n_2 \to n_1$ be the first projection, and $x = \alpha \circ \pi^{-1}$ be the pushforward measure of $\alpha$ through this projection.
The choice of name for $x$ is not a coincidence; it is the left marginal of $\alpha$: $$x(i) = \sum_{j=1}^n \alpha(i,j).$$

A \emph{disintegration} of $\alpha$ along $\pi$ is defined as a finite set $\left\{a^{(i)}\right\}_{i\in n_1}$ of measures on $n_2$ such that 
$$ \forall (i,j) \in n_1 \times n_2 : \quad \alpha(i,j) = a^{(i)}(j) \cdot x(i) $$
Equivalently, this is a kernel $n_1 \to [0,1]^{n_2}$, or a matrix in $[0,1]^{n_1\times n_2}$.
The condition for disintegration in a non-discrete setting contains an integral amounting to a finite sum here, but the measures are uniquely defined by their point masses, so this condition suffices.
In addition, the usual definition involves measurability conditions, but these are trivially satisfied here.
The kernel $A$ such that $\transpose x A = \transpose y$ is a disintegration of the joint measure $\J A = \transpose x \mem A$ defined in the previous paragraph.

\smallskip
The question is now whether $A$ is the only disintegration of $\J A$. 
The answer is yes, \emph{but only up to negligible events}. 
More precisely, given $\alpha$, the disintegration $a^{(i)}$ at $i$ -- the $i$-th row of the kernel $A$ -- can be anything if $x(i) = 0$.
It is then natural to consider an equivalence class $\equiv_x$ on kernels $A,B \in [0,1]^{n\times n}$ defined by
\begin{align*} 
  A \equiv_x B ~:\iff~ \forall i \in n~\big( x_i > 0 \implies A(i) = B(i) \big) \tag{\textup{c.f. } \eqref{eq:equivmu}}
\end{align*}
Recall that, for all $i,j$, $\J A(i,j) = A(i,j) \cdot x_i$ and $\J B(i,j) = B(i,j) \cdot x_i$.
So, we will have that 
\begin{align} \label{eq:discdisintegrate}
\J A = \J B\text{ if and only if }A \equiv_x B 
\end{align}
This is exactly the analogue of Lemma~\ref{lem:disintegrate2} below in the discrete case.

The goal of our work is to compose joint measures $\alpha,\beta \in [0,1]^{n\times n}$ with marginals $x,y\in [0,1]^n$ for $\alpha$ and marginals $y,z \in [0,1]^n$ for $\beta$, to a joint measure $\gamma \in [0,1]^{n\times n}$, without disintegrating them to kernels first.
We want this composition to satisfy some sensible properties: if $A$ and $B$ are kernels such that $\alpha = \J A = \transpose x A'$ and $\beta = \J B = \transpose y B'$, then
\begin{enumerate}
  \item the marginals of $\gamma$ are $x$ and $z$;
  \item $\gamma = \J (A \cdot B) = \transpose x (AB)'$ (c.f. Theorem~\ref{thm:embed}, $\J$ is a functor); and
  \item if $C$ is such that $\gamma = \J C$ then $AB \equiv_x C$ (faithfulness of $\J$, c.f. Lemma~\ref{lem:disintegrate2}).
\end{enumerate}
In the discrete case, the composition $\gamma$ can be defined entry-wise by
\begin{align}
  \gamma (i,k) = \sum_{\substack{j=1\\y_j>0}}^n \frac {\alpha(i,j) \cdot \beta(j,k)} {y_j} \tag{\textup{c.f. }\eqref{eq:comp}}
\end{align}
We can then verify the above properties. First, the left marginal of $\gamma$ is $x$:
$$ (\gamma \one)(i) = \sum_{k=1}^n \sum_{\substack{j=1\\y_j>0}}^n \frac{\alpha(i,j) \cdot \beta(j,k)} {y_j} = \sum_{\substack{j=1\\y_j>0}}^n \Big( \sum_{k=1}^n \beta(j,k) \Big) / y_j \cdot \alpha(i,j) = \sum_{j=1}^n \alpha(i,j) = (\alpha \one)(i) = x_i $$
Here, we have used that $\sum_{k=1}^n \beta(j,k) = y_j$ and, if, for some $j$, $y_j=0$, then $\sum_{i=1}^n \alpha(i,j) = 0$, so $\alpha(i,j)=0$ (which is why we can leave out the subscript $y_j > 0$ from the sum).
A similar calculation shows that its right marginal is $z$, so property \rom 1 is verified.
For property \rom 2, we have
$$ \J (A \cdot B) (i,k) = (A\cdot B)(i,k) \cdot x_i = \sum_{j=1}^n A(i,j)\cdot B(j,k) \cdot x_i = \sum_{\substack {j=1\\y_j>0}}^n \alpha(i,j) \cdot B(j,k) \cdot \frac {y_j} {y_j} = \sum_{\substack {j=1\\y_j>0}}^n \frac{\alpha(i,j) \cdot \beta(j,k)}{y_j} $$
identifying $\J (A \cdot B)$ with $\gamma$.
Property \rom 3 is immediate from \eqref{eq:discdisintegrate} and property \rom 2.

\section{$\Krn$ and $\JDist$}
\label{sec:notation}
We will define the basics to introduce the categories $\Krn$ and $\JDist$, though the composition of the latter will be defined in \Cref{sec:JDist} as we will need a few more results to provide it without resorting to disintegration. However, we can already give the mapping $\J \colon \Krn \to \JDist$ and prove that it is well-defined (Lemma~\ref{lem:disintegrate2}). 
\subsection{Preliminaries}
Let $(X,\AA)$ and $(Y,\BB)$ be measurable spaces. We abbreviate $(X,\AA)$ by $X$ when $\AA$ is understood. The letters $A,B,C,D$ denote measurable sets. The space $(X,\AA)$ is a \emph{standard Borel space} if $\AA$ is the set of Borel sets of a Polish space (separable and completely metrizable). The space $(X,\AA)$ is \emph{countably generated} if there exists a countable set $\AA_0\subseteq\AA$ such that $\AA$ is the smallest $\sigma$-algebra containing $\AA_0$. All standard Borel spaces are countably generated.
For an in-depth treatment of measure theory, see, e.g., \cite{Doberkat07,Halmos13,Rao87}.

A \emph{Markov kernel} $P:X\to Y$ is a map $P:X\times\BB\to[0,1]$ such that
\begin{itemize}
\item
for fixed $s\in X$, $P(s,-)$ is a probability measure on $Y$,
\item
for fixed $B\in\BB$, $P(-,B)$ is a measurable function on $X$.
\end{itemize}
These properties allow kernels to be sequentially composed by Lebesgue integration. The class of measurable spaces and Markov kernels forms a category \SRel\ \cite{Panangaden09,Doberkat07}, which is isomorphic to the Kleisli category of the Giry monad. We write $P:X\to Y$ for the kernel $P$ regarded as a morphism in this category.

For $P:X\to Y$ a Markov kernel and $\mu$ a finite measure on $X$, write $\mu\cmp P$ for the measure on $Y$ such that
\begin{align} \label{eq:measuretransform}
(\mu\cmp P)(B) &= \int_X P(s,B)\, d\mu(s). 
\end{align}
This gives a bounded linear map $(-\cmp P):\Meas X\to\Meas Y$ that is monotone and continuous in both the metric and Scott topologies~\cite{cantor}.

For $A\in\AA$, let $A$ also denote the subidentity kernel
\begin{align*}
A(s,B) &= \One_X(s,A\cap B) = \begin{cases}
1, & s\in A\cap B,\\
0, & s\not\in A\cap B.
\end{cases}
\end{align*}
Then for all $A\in\AA$ and $B\in\BB$,
\begin{align}
(\mu\cmp A\cmp P)(B) &= \int_A P(s,B)\, d\mu(s).\label{eq:kerJ}
\end{align}

The category $\Krn$\footnote{It is worth noting that the category $\Krn$ mentioned in the introduction, considered by Dahlqvist et al.~\cite{Dahlqvist18}, restricts $\AA$ to be standard Borel; we do not impose this restriction.} has as objects probability spaces $(X,\AA,\mu)$. The morphisms $(X,\AA,\mu)\to(Y,\BB,\nu)$ are equivalence classes modulo $\mu$-nullsets of Markov kernels $P\colon X\to Y$ such that
\begin{align*}
\nu(B) &= \int_{s\in X}P(s,B)\, d\mu(s).
\end{align*}

\subsection{From $\Krn$ to $\JDist$}
The category $\JDist$ is the category whose objects are probability spaces $(X,\AA,\mu)$ and whose morphisms $(X,\AA,\mu)\to(Y,\BB,\nu)$ are joint distributions or \emph{couplings} on $X\times Y$ with marginals $\mu$ and $\nu$. We will define composition in $\JDist$ formally in \Cref{sec:JDist}, but we can already define the embedding functor $\J:\Krn\to\JDist$. It is the identity on objects, and for morphisms $P:X\to Y$, let $\mem P:X\to X\times Y$ be the kernel that behaves like $P$, but also remembers its input state: on measurable rectangles $A\times B$,
\begin{align}\label{eq:markovlift}
\mem P(s,A\times B) &= \begin{cases}
P(s,B), & \text{if $s\in A$,}\\
0, & \text{if $s\not\in A$.}
\end{cases}
\end{align}
The kernel $P'$ together with the probability measure $\mu$ on $X$ induce a joint measure $\J P$ on $X\times Y$ whose value on measurable rectangles $A\times B$ is
\begin{align}
\J P(A\times B) &= \int_X \mem P(s,A\times B)\, d\mu(s) = \int_A P(s,B)\, d\mu(s).\label{eq:JPdef}
\end{align}
Since morphisms in $\Krn$ are equivalence classes of kernels, there is at this point an obligation for proof of well-definedness.
For $P,Q:X\to Y$ Markov kernels and $\mu$ a measure on $X$, the equivalence is defined as
\begin{align}\label{eq:equivmu}
P\equiv_\mu Q\ \Iff\ \forall B\in\BB\ \ \mu(\set s{P(s,B)\ne Q(s,B)})=0.
\end{align}
The following lemma now shows that $\J$ is well-defined:
\begin{lemma}
\label{lem:disintegrate2}
Let $(X,\AA)$ and $(Y,\BB)$ be measurable spaces, $\BB$ countably generated. Let $P,Q:X\to Y$ be Markov kernels and
$\mu$ a probability measure on $X$. The following are equivalent:
\begin{enumerate}[{\upshape(i)}]
\item
$P\equiv_\mu Q$
\item
$P\equiv_\xi Q$ for all $\xi$ absolutely continuous with respect to $\mu$ {\upshape(}notation: $\xi\ll\mu${\upshape)}
\item
for all $A\in\AA$, $\mu\cmp A\cmp P = \mu\cmp A\cmp Q$
\item
$\J P = \J Q$, considering $P$ and $Q$ as $\Krn$-morphisms $(X,\AA,\mu)\to(Y,\BB,\nu)$, where $\nu=\mu\cmp P=\mu\cmp Q$.
\end{enumerate}
\end{lemma}
\begin{proof}
The equivalence of (i) and (ii) is clear from the definition \eqref{eq:equivmu}.

For (i) $\Imp$ (iii), suppose that $P\equiv_\mu Q$. Let $E_B=\set s{P(s,B)=Q(s,B)}$. By definition, $\mu(X{\setminus}E_B)=0$.
For all $A\in\AA$ and $B\in\BB$,
\begin{align}
(\mu\cmp A\cmp P)(B) &= \int_A P(s,B)\, d\mu(s)\nonumber\\
&= \int_{A\cap E_B} P(s,B)\, d\mu(s)  + \int_{A{\setminus}E_B} P(s,B)\, d\mu(s)\label{eq:char1}\\
&= \int_{A\cap E_B} Q(s,B)\, d\mu(s) + \int_{A{\setminus}E_B} Q(s,B)\, d\mu(s)\label{eq:char2}\\
&= \int_{A} Q(s,B)\, d\mu(s)
= (\mu\cmp A\cmp Q)(B).\nonumber
\end{align}
The left-hand terms of \eqref{eq:char1} and \eqref{eq:char2} agree because the integral is restricted to $E_B$, and the right-hand terms are 0 because $A\setminus E_B$ is a $\mu$-nullset.
As $B$ was arbitrary, $\mu\cmp A\cmp P = \mu\cmp A\cmp Q$.

Conversely, for (iii) $\Imp$ (i), if $P\not\equiv_\mu Q$, then $\mu(\set s{\len{P(s,B)-Q(s,B)}\ge 1/n}) > 0$ for some $B\in\BB$ and $n>0$.
Letting $A$ be this set, we have
\begin{align*}
\int_{A} \len{P(s,B)-Q(s,B)}\, d\mu(s) &\ge {\textstyle\frac 1n}\mu(A) > 0,
\end{align*}
so
\begin{align*}
(\mu\cmp A\cmp P)(B) &= \int_{A} P(s,B)\, d\mu(s)
\ne \int_{A} Q(s,B)\, d\mu(s) = (\mu\cmp A\cmp Q)(B).
\end{align*}

The equivalence of (iii) and (iv) follows from \eqref{eq:kerJ} and \eqref{eq:JPdef}.
\end{proof}

Let $(X,\AA)$ be a measurable space. The countable measurable partitions $\DD$ of $X$ form an upper semilattice ordered by refinement, denoted $\sqle$, with least common refinement as join, denoted $\sqcup$. We will often consider the limiting behavior of functions defined on increasingly finer such partitions. For any such map $\phi$ taking values in a topological space, if $\phi(\DD_n)$ converges to the same value for all chains $\DD_0\sqle\DD_1\sqle\cdots$ that eventually become sufficiently fine, we write $\lim_\DD \phi(\DD)$ for this value.

\section{\RN\ Approximants}
\label{sec:LRN}

Integration and \RN\ derivatives are typically defined in terms of approximants. We will use a similar technique to define the composition of joint measures in $\JDist$. In this section we develop the necessary infrastructure.

Let $\nu$ and $\mu$ be finite measures on $(X,\AA)$. For any $B\in\AA$, consider the set
\begin{align}
\set{\fras{\nu(C)}{\mu(C)}}{C\subseteq B,\ \mu(C)>0}\subseteq \mathbb R.\label{eq:qset}
\end{align}
This set is nonempty iff $\mu(B)>0$. In that case, the set has a finite infimum, since $\nu(B)/\mu(B)$ is a member, but it may be unbounded above.

\begin{lemma}
\label{lem:approx}
Let $\nu$ and $\mu$ be finite measures on $(X,\AA)$.
For any $\eps>0$, there exists a countable measurable partition $\DD$ of $X$ such that for all $B\in\DD$ with $\mu(B)>0$, the set \eqref{eq:qset} is bounded above and
\begin{align}
\supCB - \infCB &\le \eps & 
(\supCB - \infCB)(\supCB) &\le \eps^2.\label{eq:approx2}
\end{align}
Moreover, these properties are preserved under refinement; that is, if $\DD\sqle\DD'$ and $\DD$ satisfies the inequalities \eqref{eq:approx2} for all $B\in\DD$ with $\mu(B)>0$, then the same is true of $\DD'$.
\end{lemma}
\begin{proof}
For $k\ge 1$, consider the signed measure\footnote{The signed measure $\nu-\eps k\mu$ already implies the left-hand inequality of \eqref{eq:approx2}, which suffices for most purposes, and standard treatments use this. But $\nu - (\eps\ln k)\mu$ gives the right-hand inequality as well, which will be useful in \S\ref{sec:JDist} in the definition of composition in $\JDist$.} $\nu - (\eps\ln k)\mu$. By the Hahn decomposition theorem, there exists a family of measurable partitions $\{A_k^+,A_k^-\}$ of $X$, one for each natural number $k\ge 1$, such that $\nu - (\eps\ln k)\mu$ is purely nonegative on $A_k^+$ and purely nonpositive on $A_k^-$.
That is, for all measurable $C\subseteq A_k^+$, $\nu(C) \ge (\eps\ln k)\mu(C)$ and for all measurable $C\subseteq A_k^-$, $\nu(C) \le (\eps\ln k)\mu(C)$. Without loss of generality we can assume $A_1^+=X$ and $A_1^-=\emptyset$. Let the partition $\DD$ consist of the pairwise disjoint sets $\bigcap_{i=1}^k A_i^+\cap A_{k+1}^-$, $k\ge 1$, along with $\bigcap_{i=1}^\infty A_i^+$. The last is a $\mu$-nullset, since if $\mu(C)>0$, then $C\not\subseteq A_{k}^+$ for any $k > \mathrm{exp}(\frac{\nu(C)}{\eps\mu(C)})$.

For any measurable $C\subseteq \bigcap_{i=1}^k A_i^+\cap A_{k+1}^-$, we have $(\eps\ln k)\mu(C)\le\nu(C)\le(\eps\ln(k+1))\mu(C)$, so if $\mu(C)>0$, then $\nu(C)/\mu(C)$ exists and lies in the interval $[\eps\ln k,\eps\ln(k+1)]$. Since $\ln(1+x)\le x$ for $x\ge 1$,
\begin{gather*}
\eps\ln(k+1)-\eps\ln k = \eps\ln(\frac{k+1}k) = \eps\ln(1+\frac 1k) \le \frac\eps k \le \eps\\
(\eps\ln(k+1)-\eps\ln k)\eps\ln(k+1) = \eps\ln(1+\frac 1k)\eps\ln(k+1) \le \frac\eps k\cdot\eps k = \eps^2,
\end{gather*}
from which the bounds \eqref{eq:approx2} follow.
\end{proof}

Let $\DD$ be a countable measurable partition of $\AA$. For $s\in X$, define
\begin{gather}
F_\DD = \bigcup\,\set{B\in\DD}{\mu(B)>0}\label{eq:fd}\\[1ex]
\ff_\DD^+(s) = \sum_{\substack{B\in\DD\\\mu(B)>0}}\supCB\cdot\One_X(s,B) \qquad
\ff_\DD^-(s) = \sum_{\substack{B\in\DD\\\mu(B)>0}}\infCB\cdot\One_X(s,B).\label{eq:fm}
\end{gather}
The set $F_\DD$ is measurable with $\mu(X\setcompl F_\DD)=0$, and
the functions $\ffi_\DD^+$ and $\ffi_\DD^-$ are measurable step functions that vanish outside
$F_\DD$. From Lemma \ref{lem:approx}, we have that for sufficiently fine $\DD$,
\begin{align*}
\ff_\DD^+ - \ff_\DD^- &\le \eps &
\left(\ff_\DD^+ - \ff_\DD^-\right)\ff_\DD^+ &\le \eps^2.
\end{align*}

These definitions lead to an enhanced version of the \RN\ theorem.

\begin{theorem}[Lebesgue-\RN]
\label{thm:LRN}\ 
Let $\nu$ and $\mu$ be finite measures on $(X,\AA)$.
There exist measures $\nu_0,\nu_1$, a measurable set $F\in\AA$, a measurable real-valued function $f$ defined on $X$, and a countable $\sqle$-chain $\DD_0\sqle\DD_1\sqle\cdots$ such that
\begin{enumerate}[{\upshape(i)}]
\item {\upshape(}Lebesgue decomposition{\upshape\kern.5pt)}
$\nu_0$ and $\nu_1$ form a Lebesgue decomposition of $\nu$ on $F$ with respect to $\mu$; that is,
\begin{align*}
\nu &= \nu_0+\nu_1 & \nu_0 &\ll \mu & \nu_1(F) &= 0 & \mu(X\setcompl F) &= 0;
\end{align*}
\item {\upshape(}\RN\ theorem{\upshape\kern.5pt)}
$f(s)=0$ for all $s\not\in F$ and
\begin{align}
\int_A f(s)\, d\mu(s) &= \nu_0(A),\ A\in\AA;\label{eq:RN}
\end{align}
\item {\upshape(}Uniform approximation{\upshape\kern.5pt)}
The sequence $\ffi_n^- = \ffi_{\DD_n}^-$ is monotone nondecreasing on $F$,
the sequence $\ffi_n^+ = \ffi_{\DD_n}^+$ is monotone nonincreasing everywhere, and
both sequences converge pointwise to $f$ and converge uniformly on $F$.
\end{enumerate}
\end{theorem}
If $\nu\ll\mu$, we can take $\nu_0=\nu$ and $\nu_1=0$ in (i), in which case (ii) gives $\int_A f(s)\, d\mu(s) = \nu(A)$.
In this case, $f$ is the standard \RN\ derivative $d\nu/d\mu$.

The version \cite[Theorem 3, p.~258]{Rao87} asserts (i) and (ii) only without
reference to the approximants $\ffi_n^+$ and $\ffi_n^-$,
but (iii) is essentially how it is proved. In fact, \cite{Rao87} gives three proofs. We give a fourth here for completeness.

\begin{proof}
Let $\ffi_\DD^+$, $\ffi_\DD^-$, and $F_\DD$ be defined as in \eqref{eq:fd} and \eqref{eq:fm}.
By definition,
\begin{align*}
\ff_\DD^-(s) &\le \sum_{\substack{B\in\DD\\\mu(B)>0}}\numu B\cdot\One_X(s,B)
\le \ff_\DD^+(s).
\end{align*}
If $\DD\sqle\DD'$, then $F_{\DD'} \subseteq F_\DD$, $\ffi_\DD^- \le \ffi_{\DD'}^-$ pointwise on $F_{\DD'}$, and $\ffi_{\DD'}^- \le \ffi_{\DD'}^+ \le \ffi_\DD^+$ pointwise everywhere. Moreover, by Lemma \ref{lem:approx}, for all $\eps>0$ there exists a sufficiently fine $\DD$ that $\ffi_\DD^+ - \ffi_\DD^- < \eps$ pointwise. It follows that for any countable chain $\DD_0\sqle\DD_1\sqle\cdots$ of sufficiently fine countable measurable partitions, $\ffi_n^+$ and $\ffi_n^-$ converge pointwise to a measurable function $f = \inf_n \ffi_n^+$ and converge uniformly on $F = \bigcap_n F_{\DD_n}$. In addition, the region of uniform convergence $F$ is of full $\mu$-measure, as $\mu(F)=\inf_n\mu(F_{\DD_n})=\mu(X)$, and $f$ vanishes outside $F$. This establishes (iii).

For (i), if $\mu(C)=0$, then $\nu(C\cap A_k^-) \le (\eps\ln k)\mu(C\cap A_k^-) = 0$ for all $k$, where the $A_k^-$ are the Hahn decomposition sets constructed in the proof of Lemma \ref{lem:approx}. Assuming $\DD_n$ refines one of the partitions defined in that lemma, we have $F\subseteq F_{\DD_n}\subseteq\bigcup_k A_k^-$, so $\nu(C\cap F)=0$. Taking $\nu_0(C)=\nu(C\cap F)$ and $\nu_1(C)=\nu(C\setcompl F)$ give a Lebesgue decomposition of $\nu$ on $F$ with respect to $\mu$; in particular,
\begin{align}
\mu(C)=0 &\Imp \nu(C) = \nu(C\setcompl F) = \nu_1(C).\label{eq:addN}
\end{align}

For (ii), for sufficiently fine $\DD$,
\begin{align*}
\int_A \ff_\DD^+(s)\, d\mu(s) &= \sum_{\substack{B\in\DD\\\mu(B)>0}}\supCB\cdot\mu(A\cap B)
\le \sum_{\substack{B\in\DD\\\mu(B)>0}}(\infCB+\eps)\cdot\mu(A\cap B)\\
&\le \sum_{\substack{B\in\DD\\\mu(B)>0}}(\numu{B}+\eps)\cdot\mu(B)
\le \nu(X) + \eps\mu(X),
\end{align*}
thus the integral exists and is finite. Moreover,
\begin{align}
\int_A \ff_n^+(s)\, d\mu(s)
&= \sum_{\substack{B\in\DD_n\\\mu(B)>0}}\supCB\cdot \mu(A\cap B)
= \sum_{\substack{B\in\DD_n\\\mu(A\cap B)>0}}\supCB\cdot \mu(A\cap B).\label{eq:supit}
\end{align}
The right-hand equality in \eqref{eq:supit} follows from the observation that all summands corresponding to $B\in\DD_n$ with $\mu(A\cap B)=0$ vanish, whereas for those with $\mu(A\cap B)>0$, the test $\mu(B)>0$ is redundant. Specializing the supremum in \eqref{eq:supit} at $C=A\cap B$,
\begin{align*}
\int_A \ff_n^+(s)\, d\mu(s)
&\ge \sum_{\substack{B\in\DD_n\\\mu(A\cap B)>0}}\nu(A\cap B)
= \nu(A) - \sum_{\substack{B\in\DD_n\\\mu(A\cap B)=0}}\nu(A\cap B)\\
&= \nu(A) - \sum_{\substack{B\in\DD_n\\\mu(A\cap B)=0}}\nu_1(A\cap B) & \text{by \eqref{eq:addN}}\\
&\ge \nu(A) - \sum_{\substack{B\in\DD_n}}\nu_1(A\cap B)
= \nu(A) - \nu_1(A)
= \nu_0(A).
\end{align*}
Similarly,
\begin{align*}
\int_A \ff_n^-(s)\, d\mu(s)
&= \sum_{\substack{B\in\DD_n\\\mu(A\cap B)>0}}\infCB\cdot \mu(A\cap B)
= \sum_{\substack{B\in\DD_n\\\mu(A\cap B\cap F)>0}}\infCB\cdot \mu(A\cap B\cap F).
\end{align*}
The former equality follows from an argument similar to \eqref{eq:supit}, the latter from the fact that $\mu$ vanishes outside $F$. Specializing the infimum at $C=A\cap B\cap F$,
\begin{align*}
\int_A \ff_n^-(s)\, d\mu(s)
&\le \sum_{\substack{B\in\DD_n\\\mu(A\cap B\cap F)>0}}\nu(A\cap B\cap F)
\le \sum_{\substack{B\in\DD_n}}\nu_0(A\cap B)
= \nu_0(A).
\end{align*}
Thus $\nu_0(A)$ is the limit \eqref{eq:RN}.
\end{proof}

The value of the integral \eqref{eq:RN} is independent of the choice of the $\sqle$-chain $\DD_n$, but $F$ and $f$ are not, and there is no one choice that works uniformly for all $\sqle$-chains. That is why $d\nu/d\mu$ is only defined up to a $\mu$-nullset. However, by taking least common refinements, one can find $F$ and $f$ that work for several $\sqle$-chains at once. That is one reason for making (iii) explicit.

\section{Composition in $\JDist$}
\label{sec:JDist}

To show that $\JDist$ is a category and that $\J:\Krn\to\JDist$ is a functor, we must define composition and the identity morphisms in $\JDist$ and show that they are preserved by $\J$.

Composition in $\JDist$ is defined as follows. For $\theta:(X,\AA,\mu)\to(Y,\BB,\nu)$ and $\eta:(Y,\BB,\nu)\to(Z,\CC,\xi)$, define
\begin{align}\label{eq:comp}
(\theta\cmp\eta)(A\times C)\ =\ \lim_\DD \sum_{\substack{B\in\DD\\\nu(B)>0}} \frac{\theta(A\times B)\cdot\eta(B\times C)}{\nu(B)},
\end{align}
where the limit is taken over countable measurable partitions $\DD$ of the mediating space $Y$.
We argue below (Theorem \ref{thm:arbint}) that the limit exists.

The identity morphisms are the joint distributions $\J\One_X$ obtained from the identity kernels $\One_X:X\to X$.
Checking associativity of the composition is a mechanical exercise, where one employs commutativity of the countable limit sums over the partitions of the two intermediating spaces: for $\theta$ and $\eta$ as above, and $\zeta : (Z,\CC,\xi) \to (W,\FF,\rho)$, one has
$$ (\theta\cmp\eta\cmp\zeta)(A\times F) = \lim_\DD \sum_{\substack{B\in\DD\\\nu(B)>0}} \lim_\EE \sum_{\substack{C\in\EE\\\xi(C)>0}} \frac {\theta(A\times B) \cdot \eta(B \times C) \cdot \zeta (C \times F)}{\nu(B) \cdot \xi(C)} $$
where $\DD$ and $\EE$ range over countable measurable partitions of $Y$ and $Z$, respectively.

Note that the definition of composition is completely symmetric in the input and output space. The category $\JDist$ is thus a dagger category whose involution ${}^\dagger$ is composition with transpose:
$\theta^\dagger(C\times A) = \theta(A\times C)$.

\begin{theorem}
\label{thm:arbint}
Let $\mu$, $\nu$, $\xi$ be finite measures on $Y$.
Let $\ffi_\DD^+$ and $\ggi_\DD^+$ be the approximants defined in \eqref{eq:fm}.
If there exists a countable measurable partition $\DD$ such that 
\begin{align*}
\int_Y \ff_\DD^+(t)\,\gg_\DD^+(t)\, d\mu(t) < \infty,
\end{align*}
then the limit
\begin{align}
\lim_\DD \sum_{\substack{B\in\DD\\\mu(B)>0}} \frac{\nu(B)\cdot\xi(B)}{\mu(B)}
\label{eq:arbintlimit}
\end{align}
exists and is equal to
\begin{align}
\int_Y f(t)\,g(t)\, d\mu(t) &= \inf_n \int_Y \ff_{\DD_n}^+(t)\,\gg_{\DD_n}^+(t)\, d\mu(t)
\label{eq:arbintlimit2}
\end{align}
for any sufficiently fine countable $\sqle$-chain $\DD_0\sqle\DD_1\sqle\cdots$, where
$f = \inf_n\ffi_{\DD_n}^+$ and $g = \inf_n \ggi_{\DD_n}^+$.
\end{theorem}
\begin{proof}
By definition of $\ffi_\DD^+$ and $\ggi_\DD^+$,
\begin{align*}
\int_Y \ff_\DD^+(t)\,\gg_\DD^+(t)\, d\mu(t)
&= \int_Y \left(\sumBD\supCB\One_Y(t,B)\right)\left(\sumBD\sup_{\substack{D\subseteq B\\ \mu(D)>0}}\ximu D\One_Y(t,B)\right)\, d\mu(t)\\[2ex]
&= \sumBD\left(\supCB\right)\left(\sup_{\scriptsize \substack{D\subseteq B\\ \mu(D)>0}}\ximu D\right)\,\mu(B).
\end{align*}
To show \eqref{eq:arbintlimit} and \eqref{eq:arbintlimit2}, for arbitrarily small positive $\eps$, we have by Lemma \ref{lem:approx}
\begin{align*}
\sumBD\fras{\nu(B)\xi(B)}{\mu(B)}
&= \sumBD\fras{\nu(B)\xi(B)}{\mu(B)\mu(B)}\,\mu(B)
\le \int_Y \ff_\DD^+(t)\,\gg_\DD^+(t)\, d\mu(t)\\
&\le \sumBD\left(\infCB + \eps\right)\left(\inf_{\substack{D\subseteq B\\ \mu(D)>0}}\ximu D+\eps\right)\,\mu(B)\\
&\le \sumBD\left(\numu B + \eps\right)\left(\ximu B+\eps\right)\,\mu(B)\\
&\le \sumBD\left(\fras{\nu(B)\xi(B)}{\mu(B)} + \eps\xi(B) + \eps\nu(B) + \eps^2\mu(B)\right)\\
&\le \left(\sumBD\fras{\nu(B)\xi(B)}{\mu(B)}\right) + \eps\xi(Y) + \eps\nu(Y) + \eps^2\mu(Y).
\end{align*}
For the remaining statement \eqref{eq:arbintlimit2}, we use the stronger claim \eqref{eq:approx2} of Lemma \ref{lem:approx}.
Since
\begin{align*}
& \ff_\DD^-(t) \le f(t) \le \ff_\DD^+(t) && \gg_\DD^-(t) \le g(t) \le \gg_\DD^+(t),
\end{align*}
we have
\begin{align*}
\ff_\DD^-(t)\gg_\DD^-(t) &\le f(t)g(t) \le \ff_\DD^+(t)\gg_\DD^+(t).
\end{align*}
We must choose $\DD_n$ so that
\begin{align*}
\ff_n^+(t)\gg_n^+(t) - \ff_n^-(t)\gg_n^-(t) 
\end{align*}
becomes arbitrarily small.
By \eqref{eq:approx2} of Lemma \ref{lem:approx}, $\DD$ can be chosen so that
\begin{align*}
\lefteqn{\ff_\DD^+(t)\gg_\DD^+(t) - \ff_\DD^-(t)\gg_\DD^-(t)}\qquad\\
&\le (\ff_\DD^+(t) - \ff_\DD^-(t))\gg_\DD^+(t) + \ff_\DD^+(t)(\gg_\DD^+(t) - \gg_\DD^-(t))
\le 2\eps^2.
\end{align*}
\end{proof}

\begin{corollary}
The map $\theta\cmp\eta$ as defined in \eqref{eq:comp} on measurable rectangles extends to a joint probability measure on $X\times Z$.
\end{corollary}
\begin{proof}
This can be done using the \CHK\ extension theorem\footnote{This result is sometimes referred to simply as Carath\'eodory's extension theorem and Hahn–Kolmogorov theorem, among other names.}. It suffices to verify the premises of that theorem, namely
\begin{enumerate}[{\upshape(i)}]
\item
$\theta\cmp\eta$ is finitely additive on measurable rectangles: if $\{A_n\times C_n\}_n$ is a finite set of pairwise disjoint measurable rectangles whose union is a measurable rectangle, then
\begin{align*}
(\theta\cmp\eta)(\bigcup_n(A_n\times C_n)) &= \sum_n(\theta\cmp\eta)(A_n\times C_n).
\end{align*}
\item
If for each $i\ge 0$, $\{A^i_n\times C^i_n\}_n$ is a finite collection of pairwise disjoint measurable rectangles with $\bigcup_n(A^{i+1}_n\times C^{i+1}_n)\subseteq\bigcup_n(A^i_n\times C^i_n)$, and if $\bigcap_i\bigcup_n(A^i_n\times C^i_n) = \emptyset$, then
\begin{align*}
\inf_i(\theta\cmp\eta)(\bigcup_n(A^i_n\times C^i_n)) = 0.
\end{align*}
\end{enumerate}
For (i), we can assume without loss of generality that the
$A_n$ are pairwise disjoint and the $C_m$ are pairwise disjoint,
and we are to show
\begin{align*}
(\theta\cmp\eta)(\bigcup_n A_n\times\bigcup_m C_m) &= \sum_n\sum_m(\theta\cmp\eta)(A_n\times C_m).
\end{align*}
Since limits commute with finite sums,
\begin{align*}
(\theta\cmp\eta)(\bigcup_n A_n\times\bigcup_m C_m)
&= \lim_\DD\sum_{B\in\DD}\fras{\theta(\bigcup_n A_n\times B)\cdot\eta(B\times\bigcup_m C_m)}{\nu(B)}
= \lim_\DD\sum_{B\in\DD}\sum_n\sum_m\fras{\theta(A_n\times B)\cdot\eta(B\times C_m)}{\nu(B)}\\
&= \sum_n\sum_m\lim_\DD\sum_{B\in\DD}\fras{\theta(A_n\times B)\cdot\eta(B\times C_m)}{\nu(B)}
= \sum_n\sum_m(\theta\cmp\eta)(A_n\times C_m).
\end{align*}
For (ii), 
\begin{align*}
\inf_i(\theta\cmp\eta)(\bigcup_n(A^i_n\times C^i_n))
&= \inf_i\sum_n(\theta\cmp\eta)(A^i_n\times C^i_n)\\
&= \inf_i\sum_n\lim_\DD\sum_{B\in\DD}\fras{\theta(A^i_n\times B)\cdot\eta(B\times C^i_n)}{\nu(B)}
= \inf_i\lim_\DD\sum_{B\in\DD}\sum_n\fras{\theta(A^i_n\times B)\cdot\eta(B\times C^i_n)}{\nu(B)}.
\end{align*}
We argue that if $\bigcap_i\bigcup_n(A^i_n\times C^i_n) = \emptyset$, then either $\bigcap_i\bigcup_n A^i_n = \emptyset$ or $\bigcap_i\bigcup_n C^i_n = \emptyset$. Suppose not. Let $s\in\bigcap_i\bigcup_n A^i_n$ and $t\in\bigcap_i\bigcup_n C^i_n$. Then for all $i$ there exists $n$ such that $s\in A^i_n$ and there exists $m$ such that $t\in C^i_m$. By renumbering if necessary, we can assume that for all $i$, $s\in A^i_1$ and $t\in C^i_1$. Then $(s,t)\in\bigcap_i(A^i_1\times C^i_1)\subseteq\bigcap_i\bigcup_n(A^i_n\times C^i_n)$.

By symmetry, assume without loss of generality that $\bigcap_i\bigcup_n A^i_n = \emptyset$. Since
\begin{align*}
\eta(B\times C^i_n)/\nu(B) = \eta(B\times C^i_n)/\eta(B\times Z) \le 1,
\end{align*}
we have
\begin{align*}
\inf_i\lim_\DD\sum_{B\in\DD}\sum_n\fras{\theta(A^i_n\times B)\cdot\eta(B\times C^i_n)}{\nu(B)}
&\le \inf_i\lim_\DD\sum_{B\in\DD}\sum_n\theta(A^i_n\times B)
= \inf_i\theta(\bigcup_i A^i_n\times Y)
=0.
\end{align*}
\end{proof}

\begin{theorem}[Faithfulness]\label{thm:embed}
The map $\J$ constitutes a faithful embedding $\J:\Krn\to\JDist$.
\end{theorem}
\begin{proof}
Composition in $\Krn$ is defined by $[P]_\mu\cmp[Q]_\nu = [P\cmp Q]_\mu$. This is well defined by Lemma \ref{lem:disintegrate2}.
To confirm functoriality, we show that $\J(P\cmp Q) = \J P\cmp\J Q$ and that the $\J 1_X$ are identities for composition in $\JDist$.

For composition, using the fact that for any $P$, $\J P(A\times B) = (\mu\cmp A\cmp P)(B)$, the left-hand side gives
\begin{align*}
\J(P\cmp Q)(A\times C) &= (\mu\cmp A\cmp P\cmp Q)(C).
\end{align*}
For the right-hand side, we observe that
\begin{align*}
\inf_{t\in B} Q(t,C)\int_B\nu(dt) &\le \int_B Q(t,C)\,\nu(dt)
\le \sup_{t\in B} Q(t,C)\int_B \nu(dt),
\end{align*}
or in other words,
\begin{align*}
\inf_{t\in B} Q(t,C)\,\nu(B)
&\le (\nu\cmp B\cmp Q)(C)
\le \sup_{t\in B} Q(t,C)\,\nu(B),
\end{align*}
so for $\nu(B)>0$,
\begin{align*}
\inf_{t\in B} Q(t,C) &\le \frac{(\nu\cmp B\cmp Q)(C)}{\nu(B)} \le \sup_{t\in B} Q(t,C).
\end{align*}
Thus
\begin{align*}
(\J P\cmp\J Q)(A\times C)
&= \lim_\DD \sumBD \fras{\J P(A\times B)\cdot\J Q(B\times C)}{\nu(B)}
= \lim_\DD \sumBD \fras{(\mu\cmp A\cmp P)(B)\cdot(\nu\cmp B\cmp Q)(C)}{\nu(B)}\\
&= \lim_\DD \sumBD (\mu\cmp A\cmp P)(B)\cdot \sup_{t\in B} Q(t,C)
= \int_Y (\mu\cmp A\cmp P)(dt)\cdot Q(t,C)
= (\mu\cmp A\cmp P\cmp Q)(C).
\end{align*}

For the identities, $\J 1_X\cmp\J P = \J(1_X\cmp P) = \J P$ and $\J P\cmp\J 1 = \J(P\cmp 1_Y) = \J P$.
The equivalence of Lemma \ref{lem:disintegrate2}(i) and (iv) establishes the faithfulness of the embedding. 
\end{proof}
The embedding is also full when $\Krn$ is restricted to standard Borel spaces---hence, the category considered by Dahlqvist et al.\cite{Dahlqvist18} can be fully and faithfully embedded in $\JDist$. 

\section{Conclusion}
\label{sec:pointfree}

We have presented a category of joint distributions $\JDist$ in which composition can be defined without reference to disintegration. To define this composition we explored approximants and an enhanced version of the \RN\ theorem. Our motivation to define this new category is to provide a setting in which semantics of probabilistic programs can be done under 
 weaker assumptions than previous work (where e.g. standard Borel spaces were considered). We showed that the category $\Krn$ of Markov kernels on countably generated spaces can be faithfully embedded to $\JDist$ and, when $\Krn$ is restricted to standard Borel spaces, the embedding is also full.

Although the assumption of spaces being standard Borel may not seem to be a very restrictive one, we believe our new way of composing joint measures will be profitable nevertheless.
The category $\JDist$, with its definition of composition, is a natural domain to contrive and prove correctness of approximation schemes of joint distributions for probabilistic semantics. 
This is something we already hinted at with the discrete approach in \Cref{sec:finite}.

A natural direction for future work is to explore the formulation of a point-free treatment based on the observation that the objects $(X,\AA,\mu)$ of $\Krn$ and $\JDist$ do not really depend on the measure $\mu$, but only its $\sigma$-ideal of nullsets.

\section*{Acknowledgments}

{Thanks to Fredrik Dahlqvist, Vincent Danos, Nate Foster, Justin Hsu, Bart Jacobs, Bobby Kleinberg, Radu Mardare, Prakash Panangaden, Daniel Roy, and Steffen Smolka. Special thanks to Sam Staton for catching a serious error in an earlier draft. Thanks to the Bellairs Research Institute of McGill University for providing a wonderful research environment. The support of the National Science Foundation under grant CCF-2008083 is gratefully acknowledged.} 

\bibliographystyle{entics}
\bibliography{prob,jdist}

\begin{thebibliography}{10}
\providecommand{\url}[1]{\texttt{#1}}
\providecommand{\urlprefix}{ }
\providecommand{\eprint}[2][]{\url{#2}}

\bibitem{AbramskyBlutePanangaden99}
Abramsky, S., R.~Blute and P.~Panangaden, \emph{Nuclear and trace ideals in
  tensored *-categories}, Journal of Pure and Applied Algebra \textbf{143},
  pages 3--47 (1999), ISSN 0022-4049.
\newline\urlprefix\url{https://doi.org/https://doi.org/10.1016/S0022-4049(98)00106-6}

\bibitem{ChangPollard97}
Chang, J.~T. and D.~Pollard, \emph{Conditioning as disintegration}, Statistica
  Neerlandica \textbf{51}, pages 287--317 (1997).
  \newline\urlprefix\url{ https://doi.org/10.1111/1467-9574.00056}

\bibitem{ChoJacobs17}
Cho, K. and B.~Jacobs, \emph{Disintegration and {B}ayesian inversion, both
  abstractly and concretely}, CoRR \textbf{abs/1709.00322} (2017).
  \eprint{1709.00322}.
\newline\urlprefix\url{http://arxiv.org/abs/1709.00322}

\bibitem{CulbertsonSturtz14}
Culbertson, J. and K.~Sturtz, \emph{A categorical foundation for {B}ayesian
  probability}, Applied Categorical Structures \textbf{22}, page 647–662
  (2014).
\newline\urlprefix\url{https://doi.org/10.1007/s10485-013-9324-9}

\bibitem{Dahlqvist18}
Dahlqvist, F., V.~Danos, I.~Garnier and A.~Silva, \emph{Borel kernels and their
  approximation, categorically} (2018).
\newline\urlprefix\url{https://arxiv.org/abs/1803.02651v2}

\bibitem{Doberkat07}
Doberkat, E.-E., \emph{Stochastic Relations: Foundations for Markov Transition
  Systems}, Studies in Informatics, Chapman Hall (2007), ISBN: 978-1584889410.

\bibitem{Faden85}
Faden, A.~M., \emph{The existence of regular conditional probabilities:
  necessary and sufficient conditions}, Ann. Prob. \textbf{13}, pages 288--298
  (1985).
\newline\urlprefix\url{https://projecteuclid.org/download/pdf_1/euclid.aop/1176993081}

\bibitem{Church08}
Goodman, N.~D., V.~K. Mansinghka, D.~Roy, K.~Bonawitz and J.~B. Tenenbaum,
  \emph{Church: a language for generative models}, in: \emph{Proc. 24th Conf.
  Uncertainty in Artificial Intelligence (UAI'08)}, page 220–229 (2008). ISBN: 0974903949

\bibitem{Halmos13}
Halmos, P.~R., \emph{Measure theory}, Springer (2013).
\newline\urlprefix\url{https://doi.org/10.1007/978-1-4684-9440-2}

\bibitem{Pachl78}
Pachl, J.~K., \emph{Disintegration and compact measures}, Math. Scand.
  \textbf{43}, pages 157--168 (1978).
 \newline\urlprefix\url{https://doi.org/10.7146/math.scand.a-11771}

\bibitem{Panangaden09}
Panangaden, P., \emph{Labelled Markov Processes}, Imperial College Press
  (2009).
   \newline\urlprefix\url{https://doi.org/10.1142/p595}

\bibitem{Rao87}
Rao, M.~M., \emph{Measure Theory and Integration}, Wiley-Interscience (1987), ISBN: 978-0824754013.

\bibitem{cantor}
Smolka, S., P.~Kumar, N.~Foster, D.~Kozen and A.~Silva, \emph{Cantor meets
  Scott: semantic foundations for probabilistic networks}, in: G.~Castagna and
  A.~D. Gordon, editors, \emph{Proceedings of the 44th {ACM} {SIGPLAN}
  Symposium on Principles of Programming Languages, {POPL} 2017, Paris, France,
  January 18-20, 2017}, pages 557--571, {ACM} (2017).
\newline\urlprefix\url{https://doi.org/10.1145/3009837.3009843}

\bibitem{Anglican15}
Tolpin, D., J.-W. van~de Meent and F.~Wood, \emph{Probabilistic programming in
  {A}nglican}, in: \emph{Joint European Conference on Machine Learning and
  Knowledge Discovery in Databases (ECMLKDD '15)}, pages 308--311 (2015). Available online at \url{https://link.springer.com/content/pdf/10.1007/978-3-319-23461-8_36.pdf}

\end{thebibliography}

\end{document}